\documentclass[11pt]{amsart}
\usepackage{geometry}                
\usepackage{subfig}
\usepackage{graphicx}
\usepackage{amssymb}
\usepackage{epstopdf}
\usepackage{bbm}

\usepackage[foot]{amsaddr}

\newtheorem{theorem}{Theorem}
\newtheorem{lemma}[theorem]{Lemma}
\newtheorem{corollary}[theorem]{Corollary}

\theoremstyle{remark}
\newtheorem{remark}{Remark}[theorem]

\newcommand{\R}{\mathbb{R}}

\renewcommand{\L}{\mathbf{L}}
\renewcommand{\l}{\mathbf{l}}
\renewcommand{\O}{\mathbf{\Omega}} 

\newcommand{\A}{\mathbf{A}}
\newcommand{\C}{\mathbf{C}}
\renewcommand{\v}{\mathbf{v}}
\newcommand{\q}{\mathbf{q}}
\newcommand{\Q}{\mathbf{Q}}
\newcommand{\dee}{\mathrm{d}}
\newcommand{\id}{\mathbbm{1}}
\newcommand{\ccc}{\chi}
\newcommand{\ppp}{\mathcal{P}}

\title{Twisting Somersault}
\author{Holger R.~Dullin and William Tong}
\address{The University of Sydney, School of Mathematics and Statistics}
\email{holger.dullin@sydney.edu.au, william.tong@sydney.edu.au}

\begin{document}
\maketitle

\begin{abstract}
A complete description of twisting somersaults is given using a reduction to a time-dependent
Euler equation for non-rigid body dynamics. The central idea is that after reduction the twisting 
motion is apparent in a body frame, while the somersaulting (rotation about the fixed angular 
momentum vector in space) is recovered by a combination of dynamic and geometric phase. 
In the simplest ``kick-model'' the number of somersaults $m$ and the number of twists $n$
are obtained through a rational rotation number $W = m/n$ of a (rigid) Euler top. This rotation 
number is obtained by a slight modification of Montgomery's formula \cite{Montgomery91} 
for how much the rigid body has rotated. 
Using the full model with shape changes that take a realistic time we then derive the master twisting-somersault formula:
An exact formula that relates the airborne time of the diver, the time spent in various stages of the dive, the numbers $m$ and $n$, the energy in the stages, and the angular momentum by extending a geometric phase formula due to Cabrera \cite{Cabrera07}.
Numerical simulations for various dives agree perfectly with this formula 
where realistic parameters are taken from actual observations.
\end{abstract}

\section{Introduction}

One of the most beautiful Olympic sports is springboard and platform diving, where a typical dive
consists of a number of somersaults and twists performed in a variety of forms. 
The athlete generates angular momentum at take-off and achieves the desired 
dive by executing shape changes while airborne. From a mathematical point of 
view the simpler class of dives are those for which the rotation axis and hence the direction of 
angular velocity remains constant, and only the values of the principal moments of inertia 
are changed by the shape change, but not the principal axis. This is typical in dives with
somersaults in a tight tuck position with minimal moment of inertia. 
The mathematically not so simple dives are those where the shape change 
also moves the principal axis and hence generates a motion in which the 
rotation axis is not constant. This is typical in twisting somersaults.

In this work we present a detailed mathematical theory of the twisting somersault,
starting from the derivation of the equations of motion for coupled rigid bodies in a certain co-moving frame, 
to the derivation of a model of a shape-changing diver who performs a 
twisting somersault. The first correct description of the physics of the twisting somersault 
was given by Frohlich \cite{Frohlich79}. Frohlich writes with regards to some publication 
from the 60s and 70s that
``several books written by or for coaches discuss somersaulting and twisting,
and exhibit varying degrees of insight and/or confusion about the physics processes that occur".
A full fledged analysis has been developed by 
Yeadon in a series of classical papers \cite{Yeadon93a,Yeadon93b,Yeadon93c,Yeadon93d}. 
Frohlich was the first to 
point out the importance of shape change for generating rotations even in 
the absence of angular momentum. What we want to add here is the analysis
of how precisely a shape change can generate a change in rotation in the 
presence of angular momentum. From a modern point of view this is a question 
raised in the seminal paper by Shapere and Wilczek \cite{ShapereWilczek87,ShapereWilczek88}:
``what is the most efficient way for a body to change its orientation?"
Our answer involves generalisations of geometric phase in rigid body dynamics  
\cite{Montgomery91} to shape-changing bodies recently obtained in \cite{Cabrera07}.

To be able to apply these ideas in our context we first derive a version 
of the Euler equation for a shape-changing body. Such equations are 
not new, see, e.g.~\cite{Montgomery93,Iwai98,Iwai99,Enos93}, but what is new is the  
derivation of the explicit form of the additional time-dependent terms when the 
model is a system of coupled rigid bodies. We have chosen to be less general than 
previous authors by specialising in a particular frame, but for our intended 
application the gained beauty and simplicity 
of the equations of motion (see Theorem 1) is worth the loss of generality.
We then take a particularly simple system of just two coupled rigid bodies (the ``one-armed diver'')
and show how a twisting somersault can be achieved even with this overly simplified model. 
An even simpler model is the diver with a rotor \cite{BDDLT15}, in which all the 
stages of the dive can be analytically solved for. 

Throughout the paper we emphasise the geometric mechanics point of view.
Hence the translational and rotational symmetry of the problem 
is reduced, and thus Euler-type equations are found in a co-moving frame.
In this reduced description the amount of somersault (i.e.\ the amount of rotation
about the fixed angular momentum vector in space) is not present. 
Reconstruction allows us to recover this angle by solving an additional differential equation 
driven by the solution of the reduced equations. However, it turns out that 
for a closed loop in shape space the somersault angle can be recovered by 
a geometric phase formula due to \cite{Cabrera07}. In diving terminology 
this means that from the knowledge of tilt and twist the somersault can be 
recovered. We hope that the additional insight gained in this way will lead
to an improvement in technique and performance of the twisting somersault.

The structure of the paper is as follows:
In section 2 we derive the equations of motion for a system of coupled rigid bodies that 
is changing shape. The resulting Euler-type equations are the basis for the following analysis. 
In section 3 we discuss a simplified kick-model, in which the shape change 
is impulsive. The kick changes the trajectory and the energy, but not the total 
angular momentum. In section 4 the full model is analysed, without the kick assumption.
Unlike the previous section, here we have to resort to numerics to compute
some of the terms. Nevertheless, the generalised geometric phase formula 
due to Cabrera \cite{Cabrera07} still holds and gives a beautiful geometric interpretation 
to the mechanics behind the twisting somersault.

\newpage
\section{Euler equations for coupled rigid bodies}

Let $\l$ be the constant angular momentum vector in a space fixed frame.
Rigid body dynamics usually use a body-frame because in that frame the tensor 
of inertia is constant. 
The change from one coordinate system to the other is given by a rotation matrix $R = R(t) \in SO(3)$,
such that $\l = R \L$.
In the body frame the vector $\L$ is described as a moving vector and only its length
remains constant since $R \in SO(3)$.
The angular velocity $\O$ in the body frame is the vector such that
 $ \O \times \v =  R^t \dot R \v$ for any vector $\v \in \R^3$.
Even though for a system of coupled rigid bodies the tensor of inertia is generally not 
a constant, a body frame still gives the simplest equations of motion:
\begin{theorem}
The equations of motion for a shape-changing body with angular momentum vector $\L \in \R^3$ in a body frame are
\[
         \dot \L = \L \times \O
\]
where the angular velocity $\O \in \R^3$ is given by
\[
   \O = I^{-1} ( \L - \A),
\]
$I = I(t)$ is the tensor of inertia, and $\A = \A(t)$ is a momentum shift generated by the shape change.
\end{theorem}

\begin{proof}
The basic assumption is that the shape change is such that the angular momentum is  constant. 
Let $\l$ be the vector of angular momentum in the space fixed frame, then $ \l = R \L$. 
Taking the time derivative gives $\mathbf{0} = \dot R \L + R \dot \L$ and hence $\dot \L = -R^t \dot R \L = -\O \times \L = \L \times \O$.
The interesting dynamics is all hidden in the relation between $\O$ and $\L$.

Let $\q = R \Q$ where $\Q$ is the position of a point in the body $B$
in the body frame, and $\q$ is the corresponding point in the space fixed frame.
The relation between $\L$ and $\O$ is obtained from the definition of angular momentum which is 
$\q \times \dot \q$ integrated over the body $B$.
The relation $ \q = R \Q$ is for a rigid body, for a deforming body we label each point by $\Q$ 
in the body frame but allow for an additional shape change $S$, so that $\q = R S \Q$. 
We assume that $S : \R^3 \to \R^3$ is volume preserving, which means 
that the determinant of the Jacobian matrix of $S$ is 1.
The deformation $S$ need not be linear, but we assume that we are in a frame in which the 
centre of mass is fixed at the origin, so that $S$ fixes that point. Now
\[
\begin{aligned}
    \dot \q & = \dot R S \Q + R \dot S \Q + R S \dot \Q 
             = R R^t \dot R  S \Q + R \dot S \Q = R ( \O \times S \Q )+ R \dot S S^{-1} S \Q \\
             & = R (\O \times \tilde \Q + \dot S S^{-1} \tilde \Q)
\end{aligned}
\]
where $\tilde \Q = S \Q$.
Thus we have
\[
\begin{aligned}
      \q \times \dot \q & =  R \tilde \Q \times R ( \O \times \tilde \Q + \dot SS^{-1} \tilde \Q) \\
         & = R ( | \tilde \Q|^2 \id -  \tilde \Q \tilde \Q^t) \O + R ( \tilde \Q \times \dot S S^{-1}  \tilde \Q)  \,.
\end{aligned}
\]
Now $\l$ is defined by integrating over the deformed body $\tilde B$ with density $\rho = \rho(\tilde \Q)$, 
so that $\l = \int \rho \q \times \dot \q \dee \tilde \Q$ and using $\l = R \L$ gives 
\[
     \L = \int_{\tilde B} \rho ( | \tilde \Q |{}^2 \id -  \tilde \Q \tilde \Q^t) \, \dee \tilde \Q \,\, \O + \int_{\tilde B} \rho \tilde \Q \times \dot S S^{-1} \tilde \Q \, \dee \tilde \Q.
\]
The first term is the tensor of inertia $I$ of the shape changed body and the constant term defines $\A$ so that
\[
     \L = I \O + \A
\]
as claimed.
\end{proof}

\begin{remark} 
Explicit formulas for $I$ and $\A$ in the case of a system of coupled rigid bodies are given 
in the next theorem.
When $I$ is constant and $\A = 0$  the equations reduce to the classical Euler equations for a rigid body.
\end{remark}

\begin{remark}
For arbitrary time dependence of $I$ and $\A$ the total angular momentum $| \L |$ is conserved, 
in fact it is a Casimir of the Poisson structure $\{ f, g \} = \nabla f \cdot \L \times \nabla g$.
\end{remark}

\begin{remark}
The equations are Hamiltonian with respect to this Poisson structure 
with Hamiltonian $H = \frac12 ( \L - \A) I^{-1} ( \L - \A)$ such that $\O = \partial H/ \partial \L$.
\end{remark}


For a system of coupled rigid bodies the shape change $S$ is given by rotations of the individual segments 
relative to some reference segment, typically the trunk. The orientation of the reference segment is given by the rotation matrix $R$
so that $\l = R \L$. The system of rigid bodies is described by a tree that describes the 
connectivity of the bodies, see \cite{Tong15} for the details.

Denote by $\C$ the overall centre of mass, and by $\C_i$ the position of the centre of mass of body $B_i$ relative 
to $\C$. Each body's mass is denoted by $m_i$, and its orientation by $R_{\alpha_i}$, where 
$\alpha_i$ denotes the set of angles necessary to describe its relative orientation (e.g.\ a single angle 
for a pin joint, or 3 angles for a ball and socket joint). All orientations are measured relative to the reference 
segment, so that the orientation of $B_i$ in the space fixed frame is given by $R R_{\alpha_i}$.
The angular velocity $\O_{\alpha_i}$ is the relative angular velocity corresponding to $R_{\alpha_i}$, 
so that the angular velocity of $B_i$ in the space fixed frame is $R_\alpha^t \O + \O_\alpha$.
Finally $I_i$ is the tensor of inertia of $B_i$ in a local frame with center at $\C_i$ and coordinate axes aligned with the principle axes of inertia. With this notation we have

\begin{theorem}
For a system of coupled rigid bodies we have
\[
  I = \sum R_{\alpha_i} I_i R_{\alpha_i}^t + m_i (| \C_i |{}^2 \id - \C_i \C_i^t )
\]
and
\[
\A = \sum ( m_i \C_i \times \dot \C_i + R_{\alpha_i} I_i \O_{\alpha_i} ) 
\]
where $m_i$ is the mass, $\C_i$ the position of the centre of mass, 
$R_{\alpha_i}$ the relative orientation, $\O_{\alpha_i}$ the relative angular velocity
such that $R_{\alpha_i}^t \dot R_{\alpha_i} \v = \O_{\alpha_i} \v$ for all $\v \in \R^3$, 
and $I_i$ the tensor of inertia of body $B_i$.
The sum is over all bodies $B_i$ including the reference segment, for which the rotation is simply given by $\id$. 
\end{theorem}

\begin{proof}
The basic transformation law for body $B_i$ in the tree of coupled rigid bodies 
is $\q_i = R R_{\alpha_i} ( \C_ i + \Q_i)$.
Repeating the calculation in the proof of Theorem 1 with this particular $S$ and summing over the bodies gives the result. 
We will skip the derivation of $\C_i$ in terms of the shape change and the geometry of the
model and refer to the Thesis of William Tong \cite{Tong15} for the details.
\end{proof}

\begin{remark}
In the formula for $I$ the first term is the moment of inertia of the segment transformed
to the frame of the reference segment, while the second term comes from the parallel axis theorem
applied to the centre of mass of the segment relative to the overall centre of mass.
\end{remark}

\begin{remark}
In the formula for $\A$ the first term is the internal angular momentum generated by 
the change of the relative centre of mass, while the second term originates from the relative angular velocity. 
\end{remark}

\begin{remark}
When there is no shape change then $\dot \C_i = 0$ and $\O_i = 0$, hence $\A = 0$.
\end{remark}

\begin{remark}
The vectors $\C_i$,  $\dot \C_i$, and $\O_{\alpha_i}$ are determined by the set of 
time-dependent matrices $\{ R_{\alpha_i} \}$  (the time-dependent ``shape'') and the joint positions of the coupled rigid bodies 
(the time-independent ``geometry'' of the model), see \cite{Tong15} for the details.
In particular also $\sum m_i \C_i = 0$.
\end{remark}



\section{A simple model for twisting somersault}

\begin{figure}
\centering
\subfloat[$t=0$ ]{\includegraphics[width=4cm]{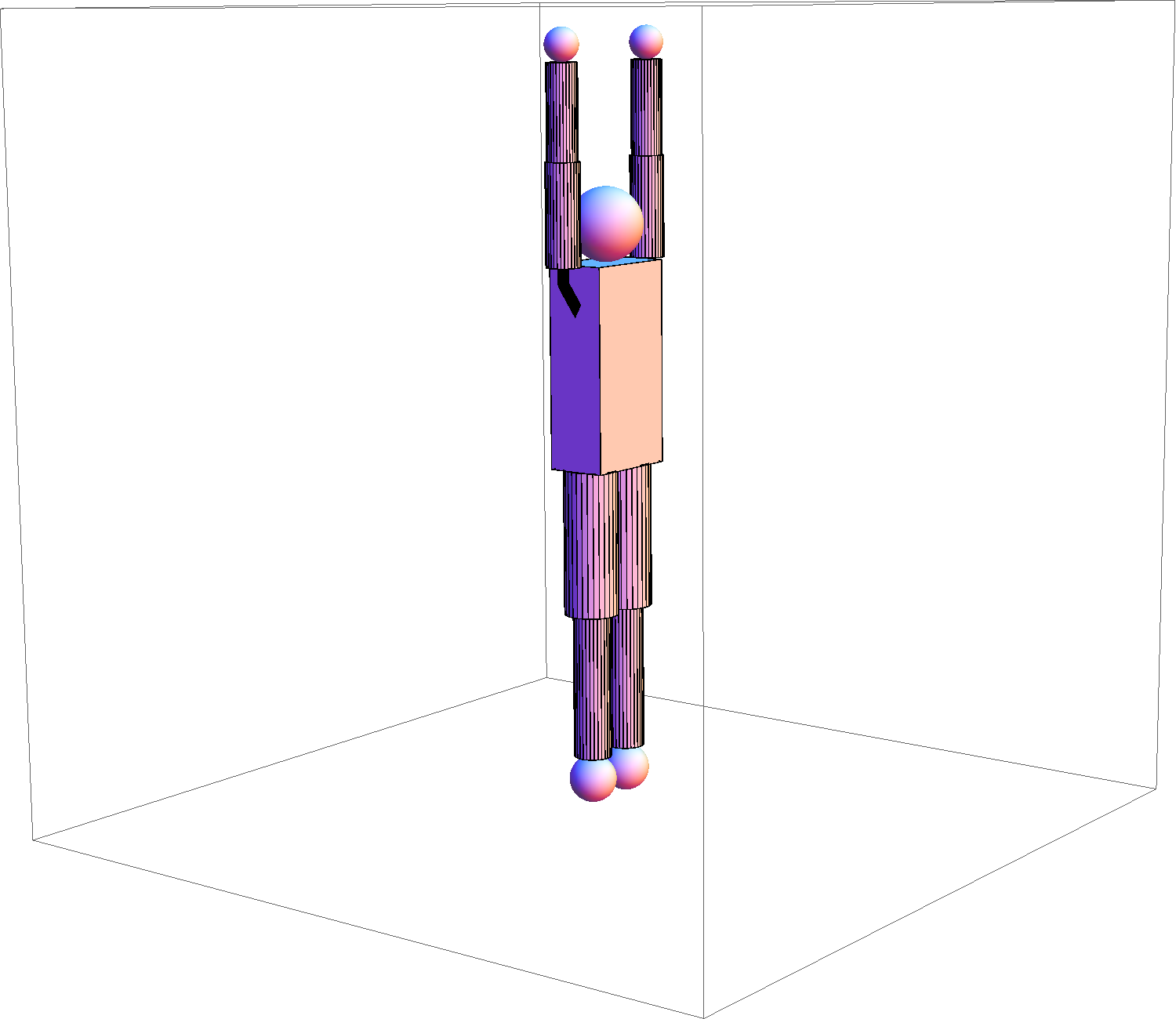}\label{fig:animationlayout}}
\subfloat[$t=1/32$ ]{\includegraphics[width=4cm]{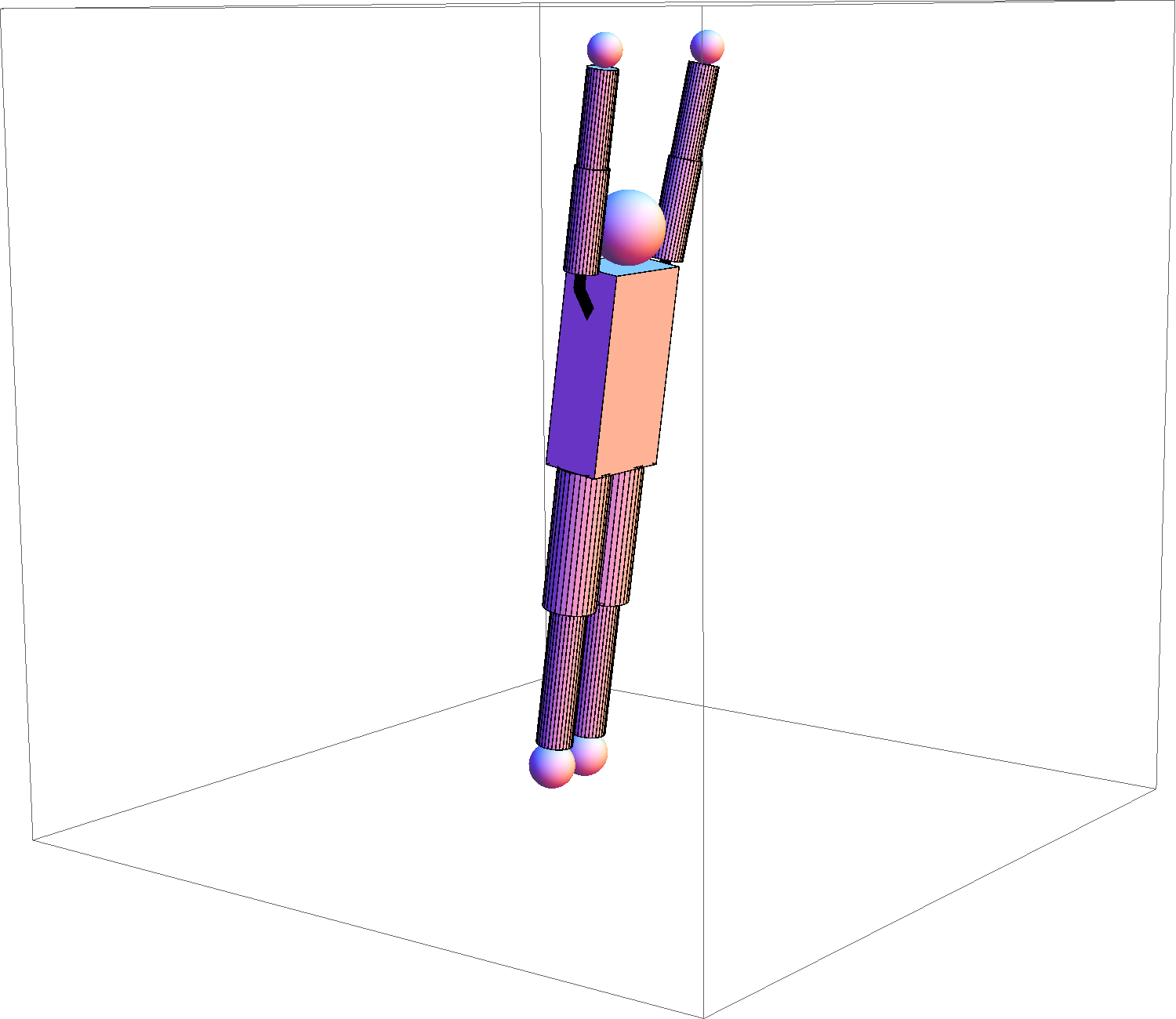}}
\subfloat[$t=1/16$ ]{\includegraphics[width=4cm]{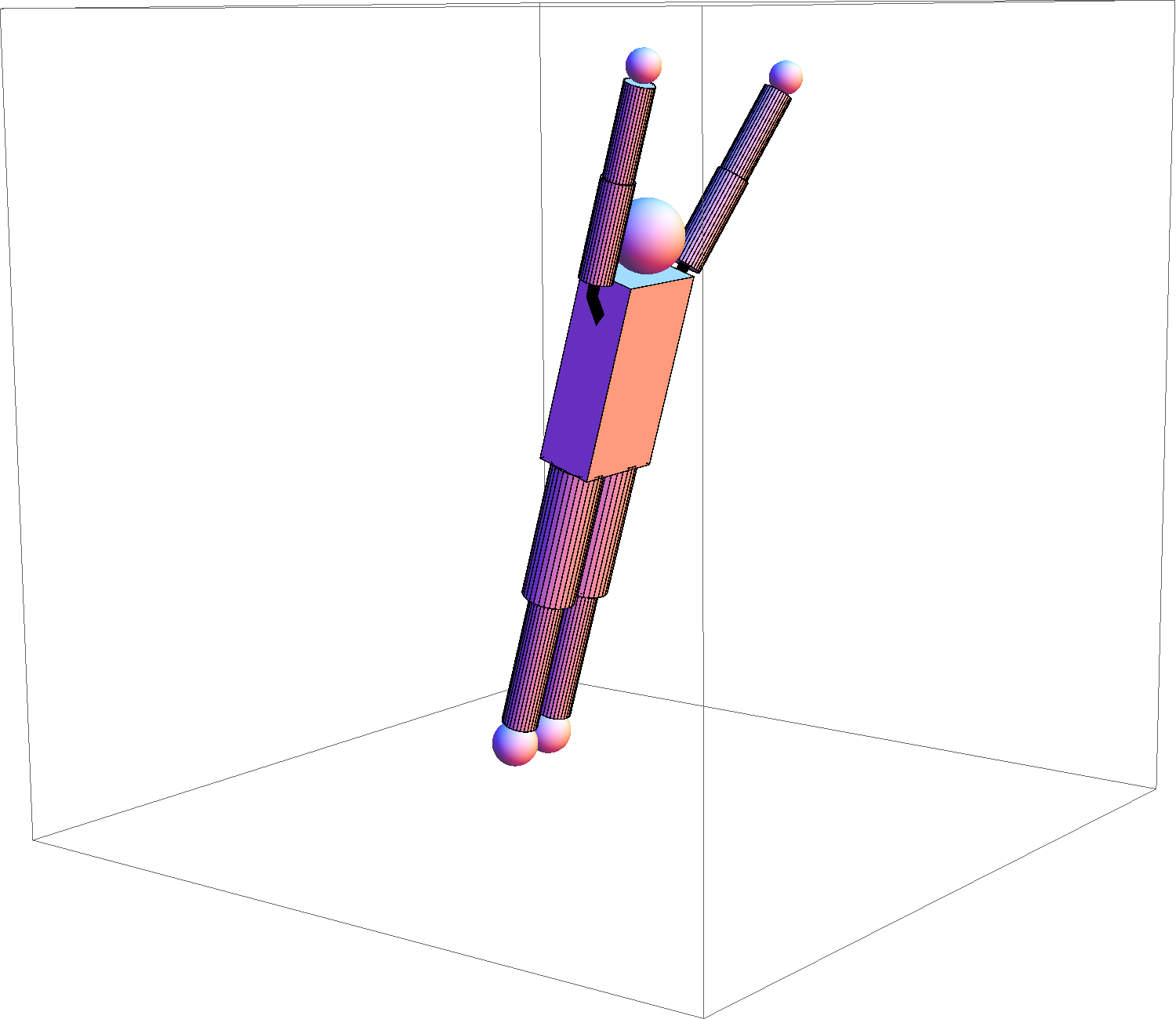}}\\
\subfloat[$t=3/32$ ]{\includegraphics[width=4cm]{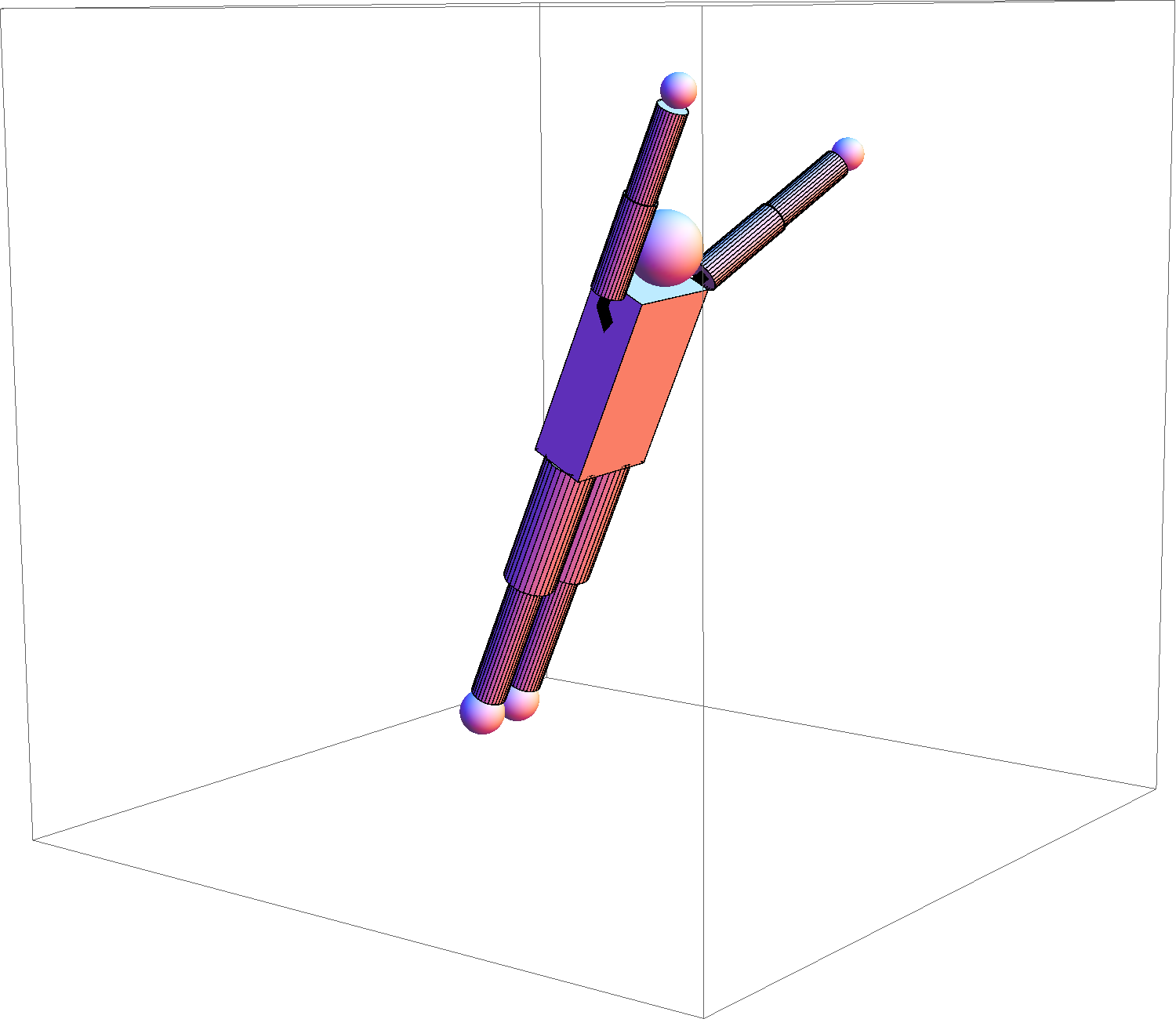}}
\subfloat[$t=1/8$ ]{\includegraphics[width=4cm]{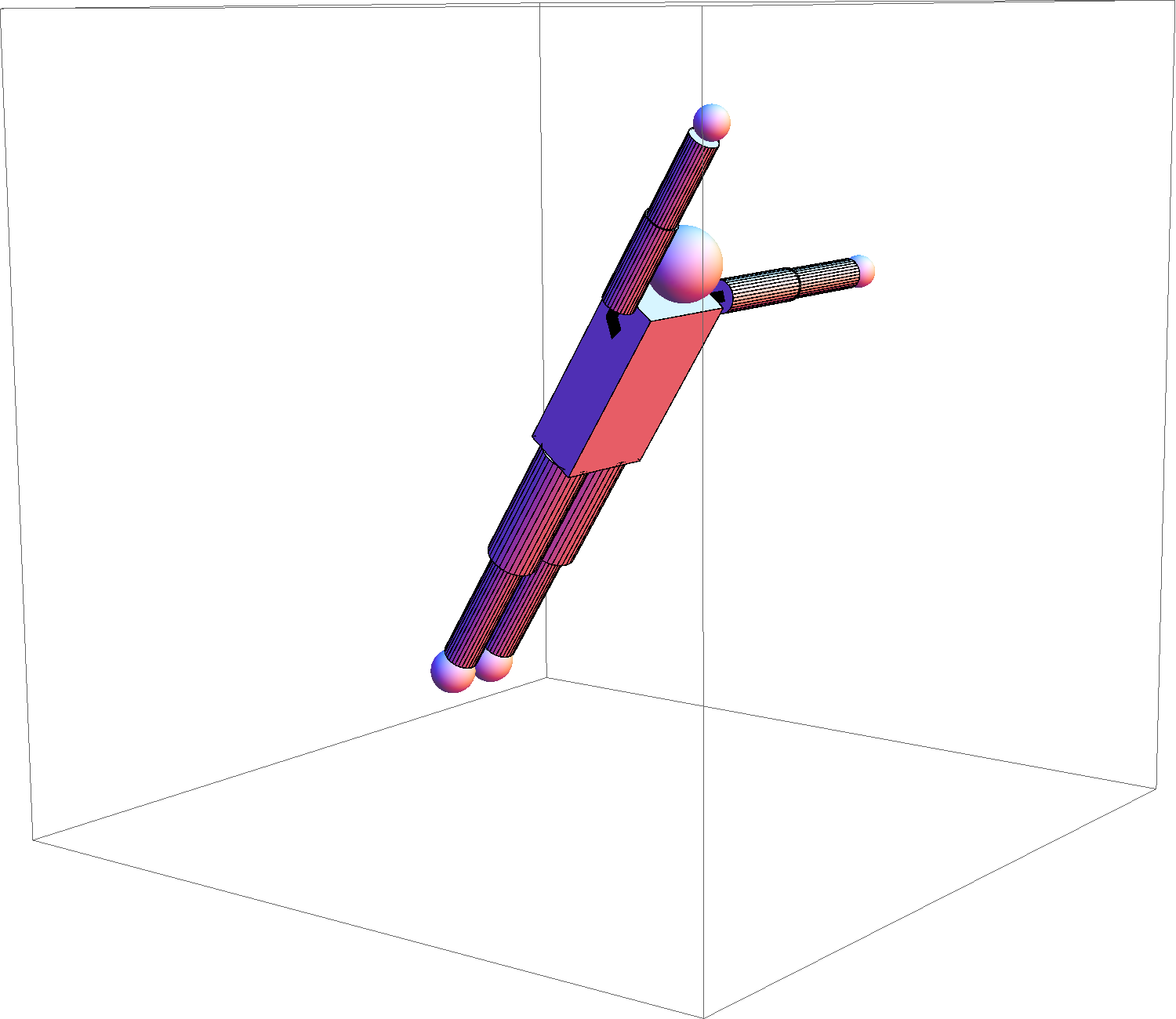}}
\subfloat[$t=5/32$ ]{\includegraphics[width=4cm]{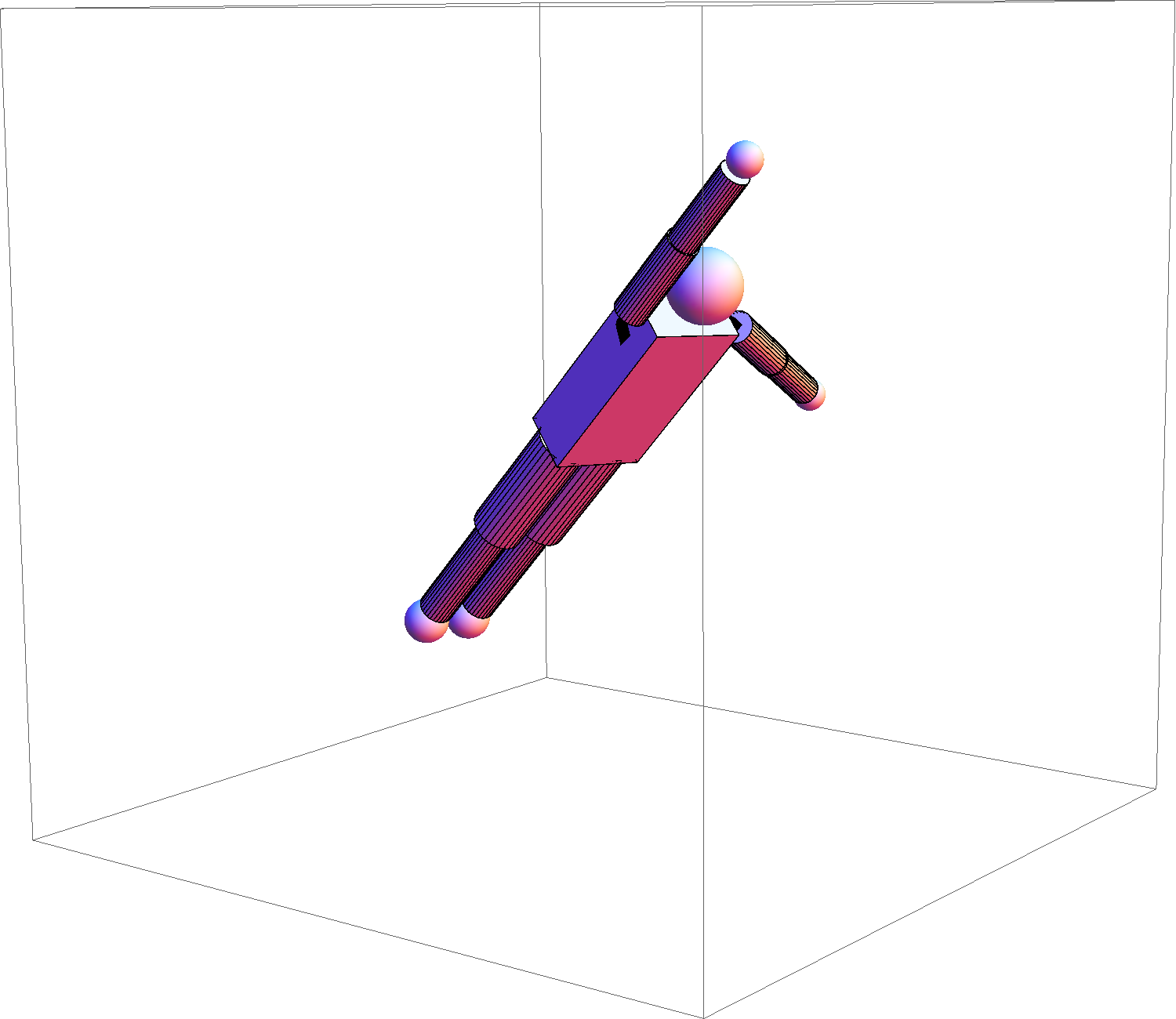}}\\
\subfloat[$t=3/16$ ]{\includegraphics[width=4cm]{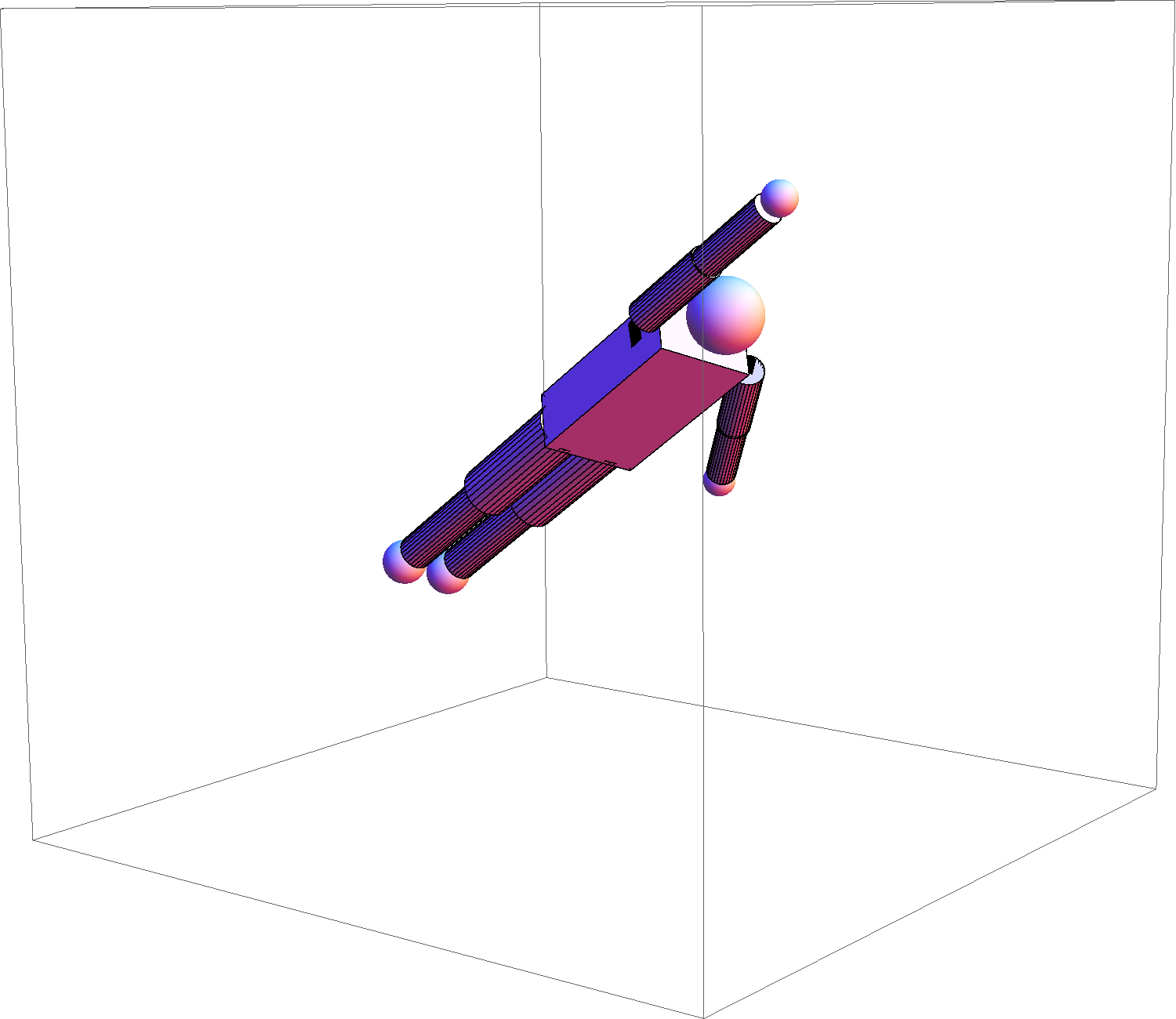}}
\subfloat[$t=7/32$ ]{\includegraphics[width=4cm]{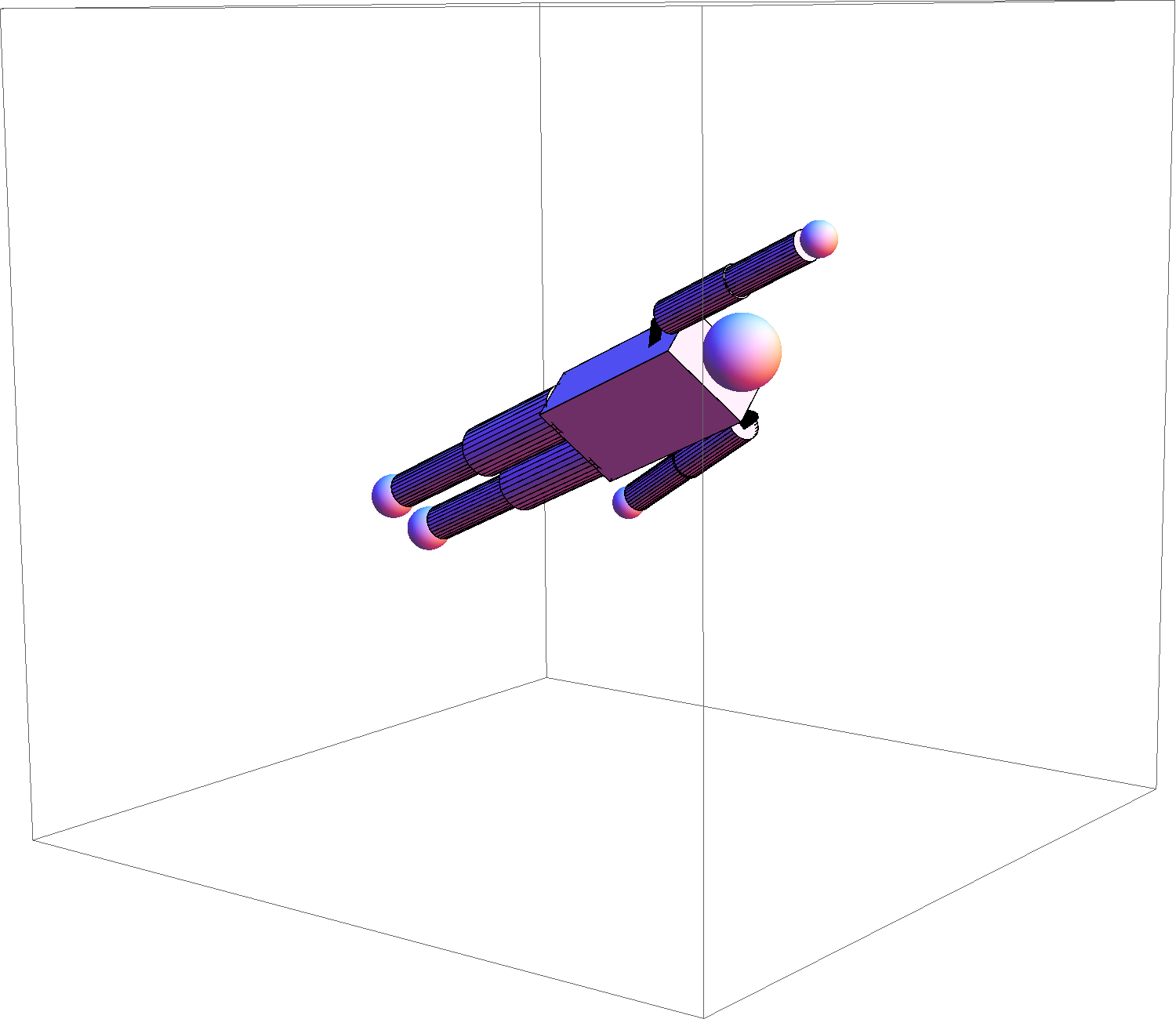}}
\subfloat[$t=1/4$ ]{\includegraphics[width=4cm]{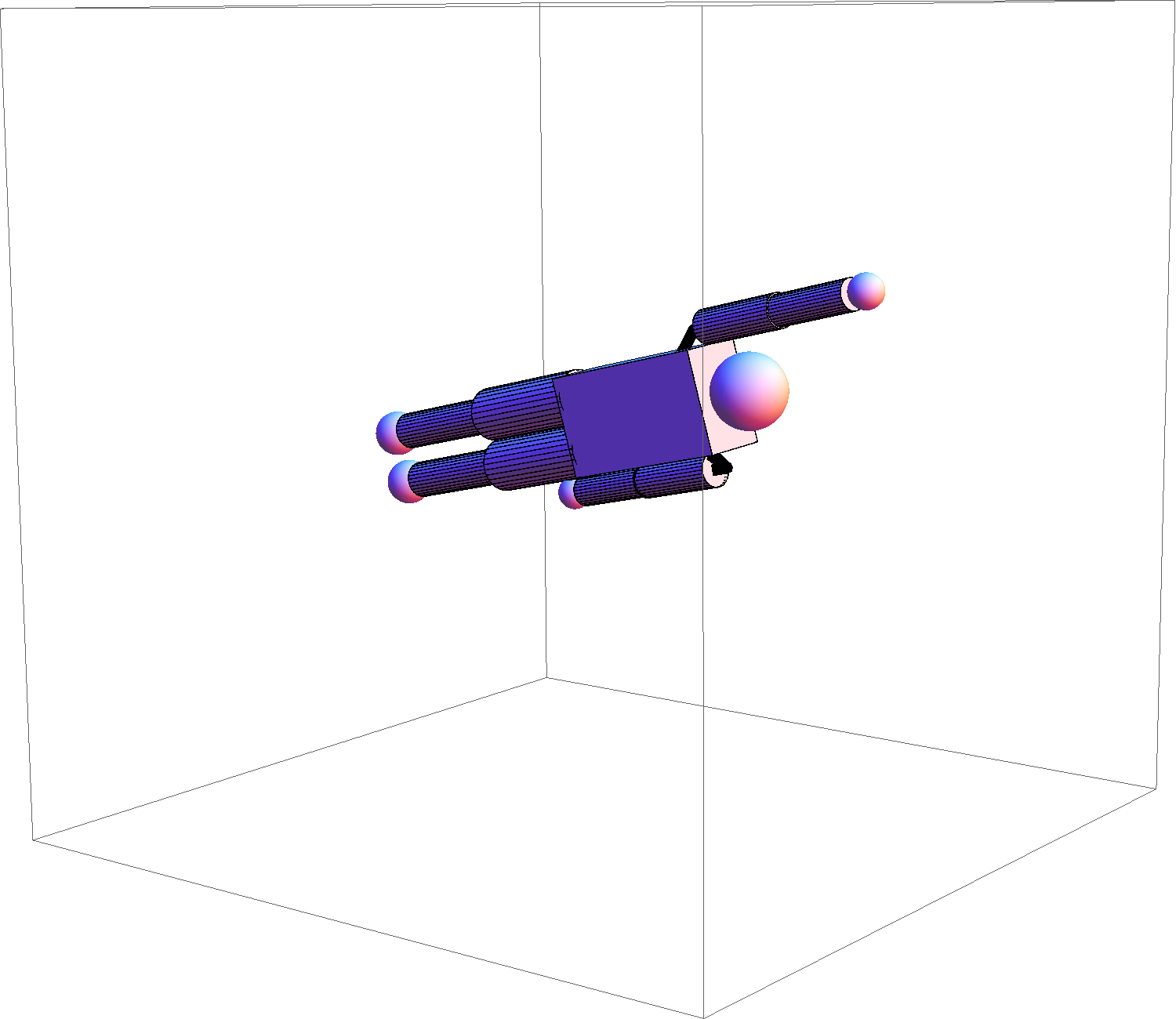}\label{fig:animationi}}
\caption{The arm motion for the twisting somersault.}\label{fig:animation2}
\end{figure}

Instead of the full complexity of a realistic coupled rigid body model for the human body,  
e.g., with 11 segments \cite{Yeadon90b} or more, 
here we are going to show that even when all but
one arm is kept fixed it is still possible to do a twisting somersault. 
The formulas we derive are completely general, so that more complicated shape 
changes can be studied in the same framework. But in order to explain the 
essential ingredients of the twisting somersault 
we chose to discuss a simple example.
A typical dive consists of a number of phases or stages in which the body shape is either fixed or not.
Again, it is not necessary to make this distinction, the equations of motion are general and one could
study dives where the shape is changing throughout. However, the assumption of  
rigid body motions for certain times is satisfied to a good approximation in reality and makes the 
analysis simpler and more explicit.
The stages where shape change occurs are relatively short, and considerable 
time is spent in rotation with a fixed shape. This observation motivates our first approximate model in 
which the shape changes are assumed to be impulsive. Hence we have 
instantaneous transitions between solution curves of rigid bodies with different tensors of inertia and 
different energy, but the same angular momentum. 
A simple twisting somersault hence looks like this: 
The motion starts out as a 
steady rotation about a principal axis resulting in pure somersault (stage 1), and typically this is about the axis of the middle principle moment of inertia which has unstable equilibrium.
After some time a shape change occurs; in our case one arm goes down (stage 2). 
This makes the body asymmetric and generates some tilt 
between the new principal axis and the constant angular momentum vector. 
As a result the body starts twisting
with constant shape (stage 3) until another shape change (stage 4) stops the twist, for which
the body then resumes pure somersaulting motion (stage 5) until head first entry in the water.
The amount of time spent in each of the five stages is denoted by $\tau_i$, where $i = 1, \dots, 5$.

\begin{figure}
\centering
\subfloat[$\L$ for all stages.]{\includegraphics[width=7.25cm]{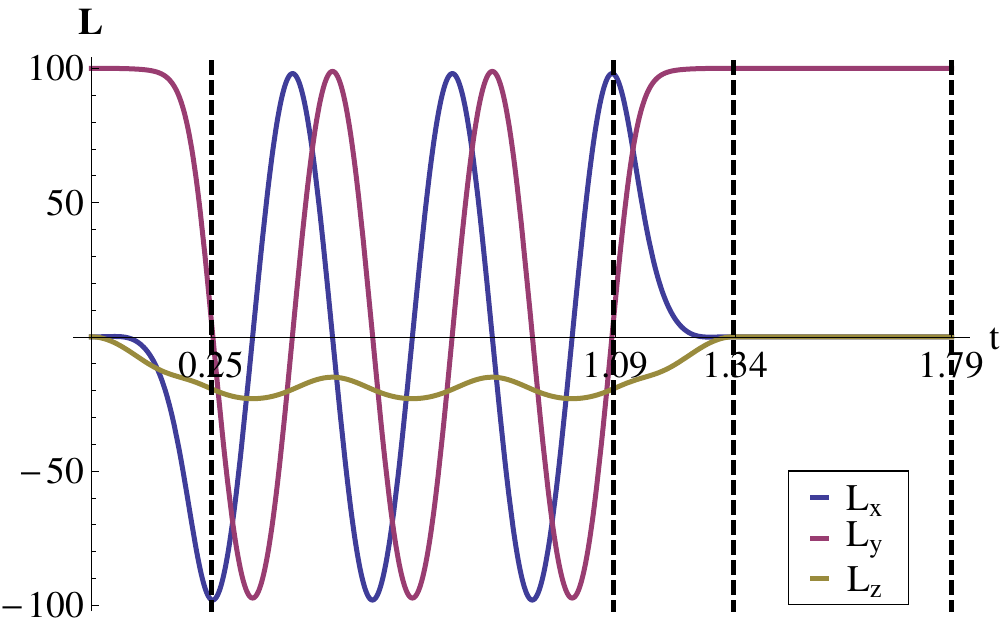}\label{fig:Lfull}}
\subfloat[$q$ for all stages.]{\includegraphics[width=7.25cm]{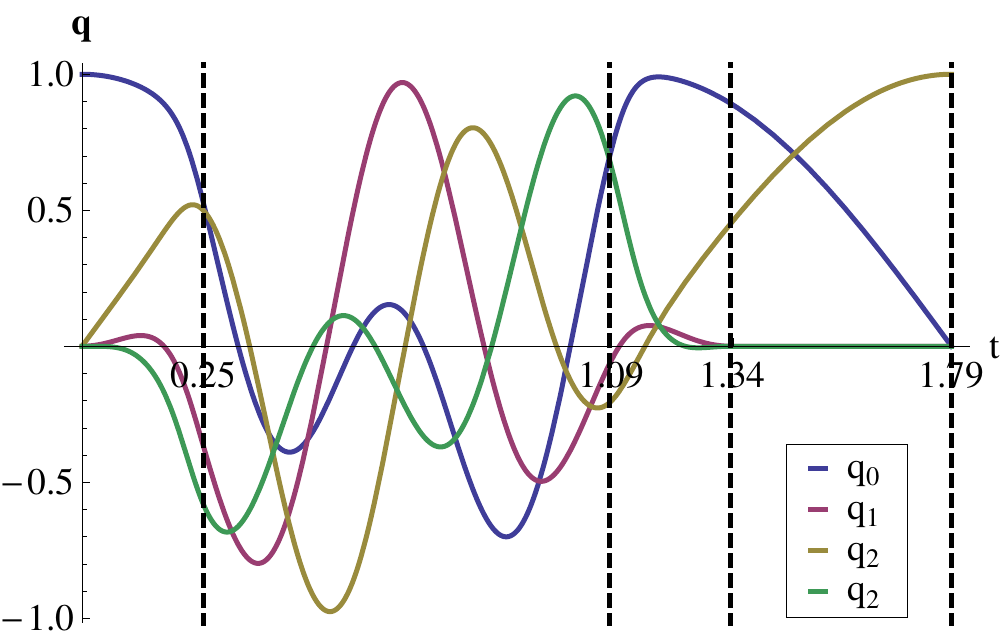}}
\caption{Twisting somersault with $m=1.5$ somersaults and $n=3$ twists.
The left pane shows the angular momentum $\L(t)$,
and the right pane the quaternion $q(t)$ that determines the orientation $R$.
The stages are separated by the vertical dashed lines, $\tau_1 = 0$, $\tau_2 = \tau_4 = 1/4$.
The same trajectory on the $\L$-sphere is shown in Fig.~\ref{fig:spacefull}.}\label{fig:Lqfull}
\end{figure}

\begin{figure}
\centering
\includegraphics[width=12cm]{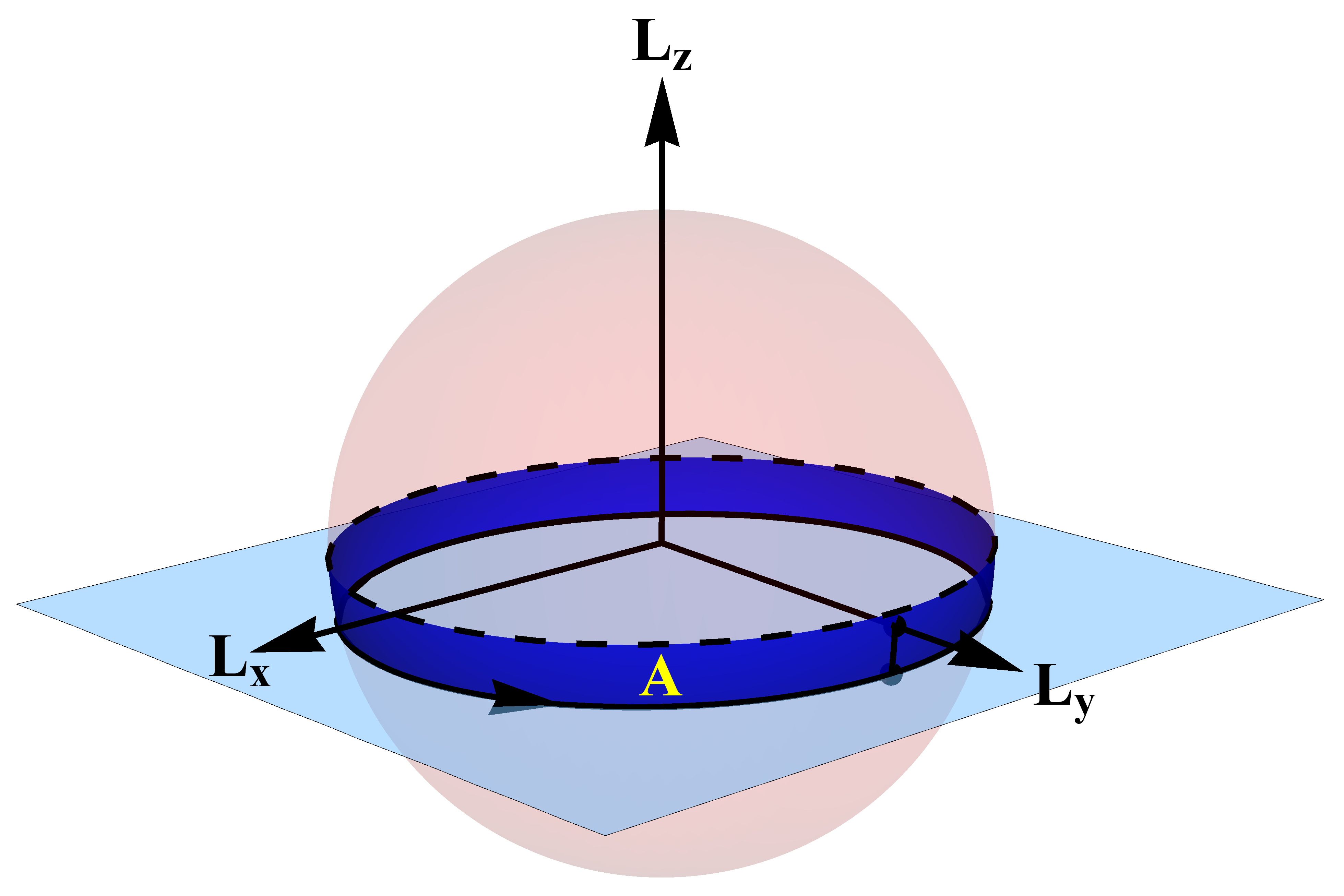}
\caption{Twisting somersault on the sphere $|\L| = l$ where the shape change is a kick.
The region $A$ bounded by the stage 3 orbit of $\L$ and the equator (dashed) is shaded in dark blue. }
\label{fig:Asimple}
\end{figure}

The energy for pure somersault in stage 1 and stage 5 is $E_s = \frac12 \L_s I_s^{-1} \L_s = \frac12 l^2 I_{s,yy}$.
In the kick-model stage 2 and stage 4 do not take up any time so $\tau_2 = \tau_4 = 0$,
but they do change the energy and the tensor of inertia. On the momentum sphere $|\L|^2 = l^2$
the dive thus appears like this, see Fig.~\ref{fig:Asimple}: For some time the orbit remains at the equilibrium point 
$L_x = L_z = 0$, then it is kicked into a periodic orbit with larger energy describing the 
twisting motion. Depending on the total available time a number of full (or half) revolutions are
done on this orbit, until another kick brings the solution back to the unstable equilibrium point 
where it started (or on the opposite end with negative $L_y$). Finally, some time is spent in the pure somersaulting motion before the athlete completes the dive with head first entry into the water. 

The description in the body frame on the sphere $|\L|^2 = l^2$ does not directly contain
the somersault rotation about $\l$ in physical space. 
This is measured by the angle of rotation $\phi$ about the fixed angular momentum vector $\l$ in space, 
which is the angle reduced by the symmetry reduction that lead to the description in the body frame, and the angle $\phi$ will have to be recovered from its dynamic and geometric phase. What is visible on the $\L$-sphere is the dynamics of the twisting motion, 
which is the rotation about the (approximately) head-to-toe body axis. 
All motions on the $\L$-sphere except for the separatrices are periodic, 
but they are not periodic in physical space
because $\phi$ in general has a different period.
This is the typical situation in a system after symmetry reduction.

To build a successful jump the orbit has to start and finish at one of the equilibrium points
that correspond to somersault rotation about the fixed momentum axis $\l$ in space. 
This automatically ensures that the number of twists performed will be a half-integer or integer.
In addition, the change in $\phi$, i.e.\ the amount of rotation about the axis $\l$ has to be such 
that it is a half-integer, corresponding to take-off in upright position and entry 
with head first into the water. If necessary the angle $\phi$ will evolve (without generating additional twist) in the pure somersaulting stages 1 and 5 to meet this condition.

The orbit for the twisting stage 3 is obtained as the intersection of the angular momentum sphere 
$|\L|^2 = l^2$ and the Hamiltonian $H = \frac12 \L I_t^{-1} \L = E_t$, where  $I_t$ denotes the 
tensor of inertia for the twisting stage 3 which in general is non-diagonal. The period of this motion depends
on $E_t$, and can be computed in terms of complete elliptic integrals of the first kind, see below.

\begin{lemma}
In the kick-model the instantaneous shape change of the arm from up ($\alpha = \pi)$ to down ($\alpha = 0$)
rotates the angular momentum vector in the body frame from 
$\L_s = (0, l, 0)^t$ to $\L_t = R_x( -\ccc) \L_s$ where 
the tilt angle is given by
\[
    \ccc = \int_0^\pi I_{t,xx}^{-1}(\alpha) A_x(\alpha) \dee \alpha  \, ,
\]
$I^{-1}_{t,xx}(\alpha)$ is the $xx$ component of the (changing) moment of inertia tensor $I_t^{-1}$, 
and the internal angular momentum is $\A = ( A_x(\alpha) \dot \alpha, 0, 0)^t$.
In particular the energy after the kick is $E_t = \frac12 \L_s R_x(\ccc) I_t R_x(-\ccc) \L_s$.
\end{lemma}
\begin{proof}
Denote by $\alpha$ the angle of the arm relative to the trunk, where $\alpha = 0$ is arm down 
and $\alpha = \pi$ is arm up. Let the shape change be determined by $\alpha(t)$ where $\alpha(0) = \pi$ 
and $\alpha(\tau_2) = 0$. Now $\A$ is proportional to $\dot \alpha$ and hence diverges 
when the time that it takes to do the kick goes to zero, $\tau_2 \to 0$.
Thus we approximate $\O = -I^{-1} \A$, and hence the equations 
of motion for the kick becomes
\[
   \dot \L = I^{-1} \A \times \L
\]
and are linear in $\L$. 
Denote the moving arm as the body $B_2$ with index 2, and the trunk and all the other fixed segments as a combined body with index 1.
Since the arm is moved in the $yz$-plane we have $R_{\alpha_2} = R_x(\alpha(t))$ and  $\Omega_{\alpha_2}$ 
parallel to the $x$-axis. Moreover the overall centre of mass
will be in the $yz$-plane, so that $\C_i$ and $\dot \C_i$, $i=1,2$, are also in the $yz$-plane. So we have
$\C_i \times \dot \C_i$ parallel to the $x$-axis as well, and hence $\A = (A_x \dot\alpha, 0, 0)^t$.
The parallel axis theorem gives non-zero off-diagonal entries only in the $yz$-component of $I$,  
and similarly for $R_{\alpha_2} I_s R^t_{\alpha_2}$; hence $I_{xy} = I_{xz} = 0$ for this shape change.
Thus $I^{-1} \A = ( I_{t,xx}^{-1} A_x \dot \alpha, 0, 0)^t$, 
and the equation for $\dot \L$ can be written as $\dot \L = f(t) M \L$ where $M$ is a constant matrix 
given by $M = \frac{\dee}{\dee t} R_x(t)|_{t=0}$ and $f(t) =  I_{t,xx}^{-1}(\alpha(t)) A_x(\alpha(t)) \dot \alpha$.
Since $M$ is constant we can solve the time-dependent linear equation and find 
\[
   \L_t = R_x( -\ccc) \L_s
\]
(the subscript $t$ stands for ``twist", not for ``time") where 
\[
    \ccc = \int_0^{\tau_2} I_{t,xx}^{-1} A_x \dot \alpha \,\dee t = \int_0^\pi I_{t,xx}^{-1}(\alpha) A_x(\alpha)\, \dee \alpha  \,.
\]
We take the limit $\tau_2 \to 0$ and obtain the effect of the kick, which is a change in $\L_s$ by a rotation 
of $-\ccc$ about the $x$-axis. The larger the value of $\ccc$ the higher the twisting orbit is on the sphere,
and thus the shorter the period. 
The energy after the kick is easily found
by evaluating the Hamiltonian at the new point $\L_t$ on the sphere with the new tensor of inertia $I_t$.
\end{proof}

The tensor of inertia after the shape change is denoted by $I_t$, it does not change by much in comparison to $I_s$ but is now non-diagonal, however, with the rotation $R_x(\ppp)$ 
it can be re-diagonalised to 
\[
   J = \mathrm{diag}( J_x, J_y, J_z) = R_x(-\ppp) I_t R_x(\ppp)
\]
where in general the eigenvalues are distinct. The precise formula for $\ppp$ depends on the inertia properties
of the model, but the value for a realistic model with 10 segments it can be found in \cite{Tong15}.
Formulas for $A_x(\alpha)$ in terms of realistic inertial parameters are also given in \cite{Tong15}.

In stage 3 the twisting somersault occurs, where we assume an integer number of twists occur. 
This corresponds to an integer number of revolutions of the periodic orbit of the rigid body with 
energy $E_t$ and tensor of inertia $I_t$. Let $T_t$ be the period of this motion.
As already pointed out the amount of rotation about the fixed angular momentum vector $\l$ cannot 
be directly seen on the $\L$-sphere. Denote this angle of rotation by $\phi$. We need the total change  $\Delta \phi$ 
to be an odd multiple of $\pi$ for head-first entry. 
Following Montgomery \cite{Montgomery91} the amount $\Delta \phi$ can be split into a 
dynamic and a geometric phase, where the geometric phase is given by the solid angle $S$
enclosed by the curve on the $\L$-sphere.
We are going to re-derive the formula for $\Delta \phi$ here using simple 
properties of the integrable Euler top. 
The formula for the solid angle $S$ enclosed by a periodic orbit on the $\L$-sphere 
is given in the next Lemma (without proof, see, e.g.~\cite{Montgomery91,Levi93,Cushman05}).
\begin{lemma} \label{lem:Ssphere}
The solid angle on the sphere $\L^2 = l^2$ enclosed by the intersection with the ellipsoid $\L J^{-1} \L = 2 E$ is given by 
\[
   S(h, \rho) =  \frac{4 h  g}{\pi}  \Big( \Pi(n,m) -  K(m)\Big) 
\]
where 
$m = \rho ( 1 - 2 h \rho)/( 2 h + \rho)$, 
$n = 1 - 2 h \rho$, 
$g = ( 1 + 2 h /\rho)^{-1/2}$, 
$h = ( 2 E_t J_y/l^2 - 1)/ \mu$,
$\rho = (1 - J_y/J_z)/(J_y/J_x - 1)$, 
$\mu = (J_y/J_x - 1)(1 - J_y/J_z)$.
\end{lemma}

\begin{remark}
The essential parameter is $E_t J_y/l^2$ inside $h$, and  all other dependences are on certain dimensionless
combinations of moments of inertia. 
Note that the notation in the Lemma uses the classical letter $n$ and $m$ for the arguments 
of the complete elliptic integrals, these are not to be confused with the number of twists and 
somersaults.
\end{remark}

\begin{remark}
This is the solid angle enclosed by the curve towards the north-pole 
and when measuring the solid angle between the equator and the curve the results is $2 \pi - S$.
This can be seen by considering $h \to 1/( 2 \rho)$ which implies $m \to 0$ and $n \to 0$, 
and hence $S \to 0$.
\end{remark}

Using this Lemma we can find simple expressions for the period and rotation number by 
noticing that the action variables of the integrable system on the $\L$-sphere
are given by $l S/\pi$:
\begin{lemma}\label{lem:tT}
The derivative of the action $l S /( 2 \pi) $ with respect to the energy $E$ gives the inverse of the frequency 
$(2\pi)/T$ of the motion, such that the period $T$ is
\[
 T = \frac{ 4 \pi g}{\mu l}  K(m) \,. 
\]
and the derivative of the action $lS/( 2\pi) $ with respect to $l$ gives the rotation number $-W$.
\end{lemma}
\begin{proof}
The main observation is that the symplectic form on the $\L$-sphere of radius $l$ is the area-form 
on the sphere divided by $l$, and that the solid angle on the sphere of radius $l$ is the enclosed area divided by $l^2$.
Thus the action is $ l S / ( 2\pi)$ and the area is $l^2 S$. 
From a different point of view the reason that the essential object is the solid angle is that the 
Euler-top has scaling symmetry: if $\L$ is replaced by $s \L$ then  $E$ is replaced by $s^2 E$
and nothing changes. This implies that the essential parameter is the ratio $E/l^2$ and the solid 
angle is a function of $E/l^2$ only. 
A direct derivation of the action in which $h$ is the essential parameter can be found in  \cite{PD12}.
Now differentiating the action $l S(E/l^2) / ( 2\pi) $ with respect to $E$ gives 
\[
  \frac{ T}{2\pi} =  \frac{ \partial l S(E/l^2) }{2\pi \partial E} = \frac{1}{2\pi l} S'( E/l^2)
\]
and differentiating the action with respect to $l$ gives 
\[
   -2\pi W =  \frac{ \partial l S(E/l^2) }{\partial l} = S(E / l^2) - \frac{2 E}{l^2} S'(E/l^2)  = S - \frac{ 2 E T}{ l}\vspace{-7mm}
\]
\end{proof}

\vspace{0mm}
\begin{remark}
What is not apparent in these formulas are that the scaled period $l T$ is a relatively simple
complete elliptic integral of the first kind (depending on $E/l^2$ only),
while $S$ and $W$ are both complete elliptic integrals of the third kind
(again depending on $E/l^2$ only).
\end{remark}

\begin{remark}
In general the rotation number is given by the ratio of frequencies, and those can be computed 
from derivatives of the actions with respect to the integrals. If one of the integrals (of the original 
integrable system) is a global $S^1$ action, then the simple formula $W = \partial I / \partial l$
results, which is the change of the action $I$ with respect to changing the other action $l$
while keeping the energy constant. 
\end{remark}


\begin{theorem}
The total amount of rotation $\Delta \phi_{kick}$ about the fixed angular momentum axis~$\l$ for the kick-model
when performing $n$ twists is given by 
\begin{equation} \label{eqn:kick}
    \Delta \phi_{kick} = (\tau_1 + \tau_5) \frac{2 E_s }{l} + \tau_3 \frac{2 E_t}{l} - n S\,.
\end{equation}
The first terms are the dynamic phase where $E_s$ is the energy in the somersault stages and
$E_t$ is the energy in the twisting somersault stage.
The last term is the geometric phase where $S$ is the solid angle enclosed by the orbit in the twisting somersault 
stage. For equal moments of inertia $J_x = J_y$ the solid angle $S$ is
\[
 S = 2\pi \sin( \mathcal{\ccc} + \ppp) 
\]
and in general is given by Lemma~\ref{lem:Ssphere}.
To perform $n$ twists the time necessary is $\tau_3 = n T_t$ where
again for equal moments of inertia $J_x = J_y$ the period $T_t$ is
\[
   T_t = \frac{2\pi}{l} \frac{(J_y^{-1} - J_z^{-1})^{-1}}{ \sin ( \ccc + \ppp )}  
\]
and in general is given by Lemma~\ref{lem:tT} where $T=T_t$.
\end{theorem}
\begin{proof}
This is a slight extension of Montgomery's formula \cite{Montgomery91} for how much the rigid body rotates, with 
the added feature that there is no $\mathrm{mod}  \, 2\pi$ for $\Delta \phi_{kick}$: We actually need to know how many somersaults occurred.
The formula is applied to each stage of the dive, but stage 2 and stage 4 do not contribute because $\tau_2 = \tau_4 = 0$.
In stage 1 and stage 5 the trajectory is at an equilibrium point on the $\L$-sphere, so there is only a contribution to the dynamic phase. 
The essential terms come from stage 3, which is the twisting somersault stage without shape change. 
When computing $S$ we need to choose a particular normalisation of the integral which is different from 
Montgomery \cite{Montgomery91,Levi93}, and also different from  \cite{Cushman05}. 
Our normalisation is such that when 
$J_x = J_y$ the amount of rotation obtained is the corresponding angle $\phi$ of the somersault, 
i.e.\ the rotation about the fixed axis $\l$ in space. This means that the correct solid angle for our purpose is 
such that when $J_x = J_y$ and $E_t = E_s$ the contribution is zero. Therefore, we should measure
area $ A = S l^2$ relative to the equator on the sphere. When $J_x = J_y$ we are simply measuring the area
of a slice of the sphere bounded by the equator and the twisting somersault orbit, which the latter is in a 
plane parallel to the $xy$-plane with opening angle $\ccc + \ppp$, see Fig.~\ref{fig:Asimple}.
\footnote{This is somewhat less obvious then it appears, since the orbit is actually tilted by $\ppp$ 
relative to the original equator. It turns out, however, that computing it in either way gives the same answer.}
In the general case where $J_x \neq J_y$ the area can be computed in terms of elliptic integrals, and the details are given in \cite{Tong15}.
Similarly the period of the motion along $H = E_t $ can be either computed from explicit solutions
 of the Euler equations for $J_x = J_y$, or by elliptic integrals, again see \cite{Tong15} for the details.
\end{proof}

Now we have all the information needed to construct a twisting somersault.
A result of the kick approximation is that we have $\tau_2 = \tau_4 = 0$, and if we further set $\tau_1=\tau_5=0$ then there is no pure somersault either, which makes this the simplest twisting somersault. We call this dive the pure twisting somersault and take it as a first approximation to understanding the more complicated dives.

\begin{corollary}
A pure twisting somersault with $m$ somersaults and $n$ twists is found for
$\tau_1 = \tau_2 = \tau_4 = \tau_5 = 0$ and must satisfy
\begin{equation} \label{eqn:rot}
    2 \pi m = \Big( 2   l T_t  \frac{ E_t}{l^2} -  S \Big) n 
\end{equation}
where both $S$ and $l T_t$ are functions of $E_t/l^2$ only (besides inertial parameters).
\end{corollary}
\begin{proof}
This is a simple consequence of the previous theorem by setting $\Delta \phi = 2 \pi m$, $\tau_3 = n T_t$,
and $\tau_1 = \tau_5 = 0$.
\end{proof}

\begin{remark} 
Solving \eqref{eqn:rot} for $m/n$ gives a rotation number of the Euler top, which characterizes 
the dynamics on the 2-tori of the super-integrable Euler top. This rotation number is equivalent 
up to uni-modular transformations to that of Bates et al \cite{Cushman05}.
\end{remark}

\begin{remark}
The number of somersaults per twists is $m/n$, and \eqref{eqn:rot} determines $E_t/l^2$ (assuming the inertial parameters are given).
Having $E_t/l^2$ determined in this way means one would need to find a shaped change or kick which achieves that $E_t/l^2$, and
large values of $E_t/l^2$ can be hard or impossible to achieve.
For the one-arm kick-model discussed above the energy that is reached is given by 
\[
    \frac{E_t}{l^2} = \frac12 \L_s R_x(\ccc + \ppp) J R_x(-\ccc - \ppp) \L_s/l^2. 
\]
\end{remark}

\begin{remark}
Given a particular shape change (say, in the kick-approximation) the resulting $E_t/l^2$ will 
in general not result in a rational rotation number, and hence not be a solution of \eqref{eqn:rot}.
In this case the pure somersault 
of stage 1 and/or stage 5
needs to be used to achieve a solution of \eqref{eqn:kick}
instead.
\end{remark}

\begin{remark}
The signs are chosen so that $S$ is positive in the situation we consider. Thus the geometric phase 
lowers $\Delta \phi$, and can be thought of as an additional cost that twisting on for somersaulting. 
\end{remark}

\begin{figure}
\centering
\includegraphics[width=10cm]{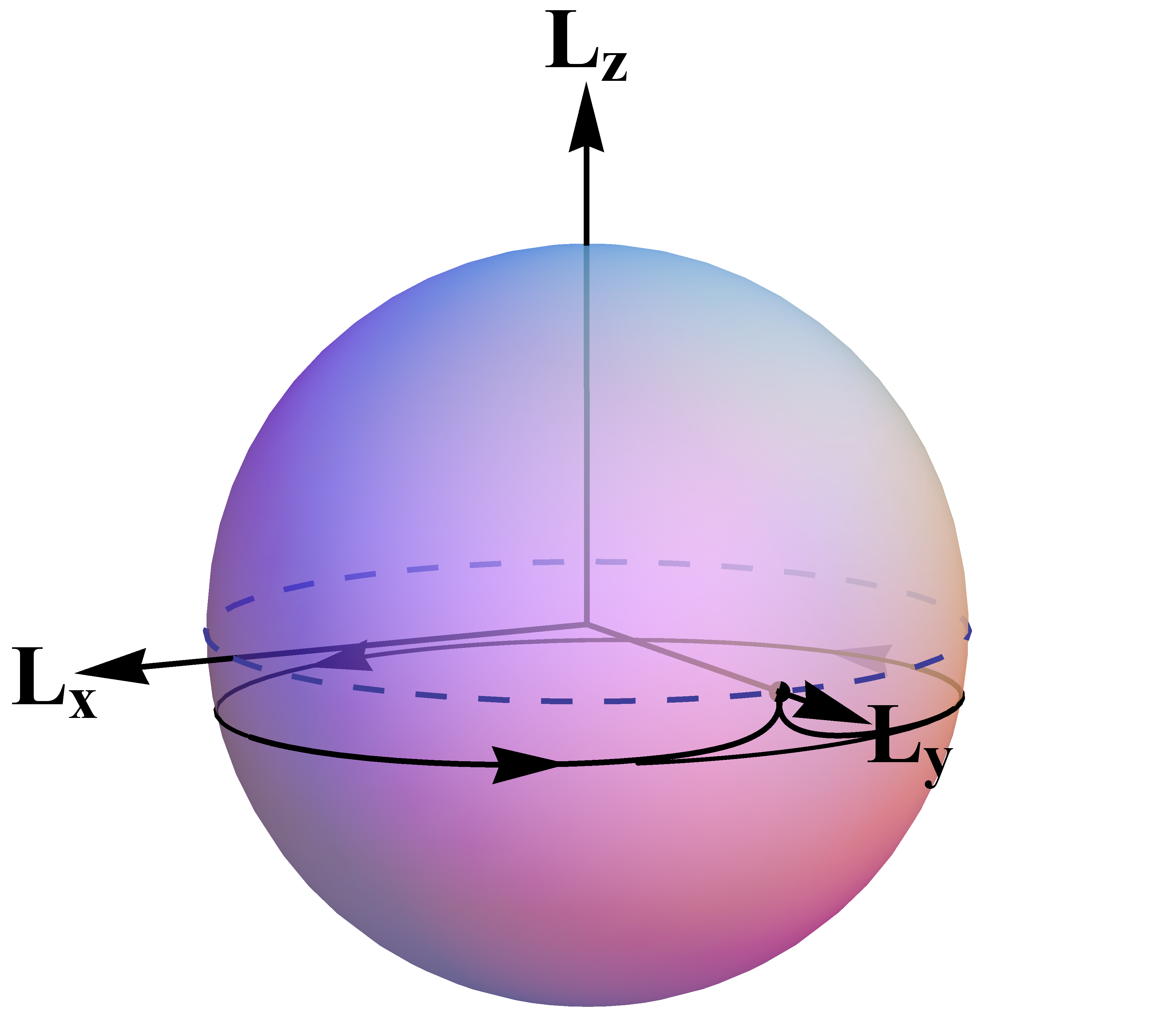}
\caption{Twisting somersault with $m=1.5$ somersault and $n=3$ twists on the sphere $|\L| = l$.
The orbit starts and finishes on the $L_y$-axis with stage 1 and stage 5. 
Shape-changing stages 2 and 4 are the curved orbit segments that start and finish 
at this point. The twisting somersault stage 3 appears as a slightly deformed 
circle below the equator (dashed).
}\label{fig:spacefull}
\end{figure}

\begin{figure}
\centering
\includegraphics[width=7cm]{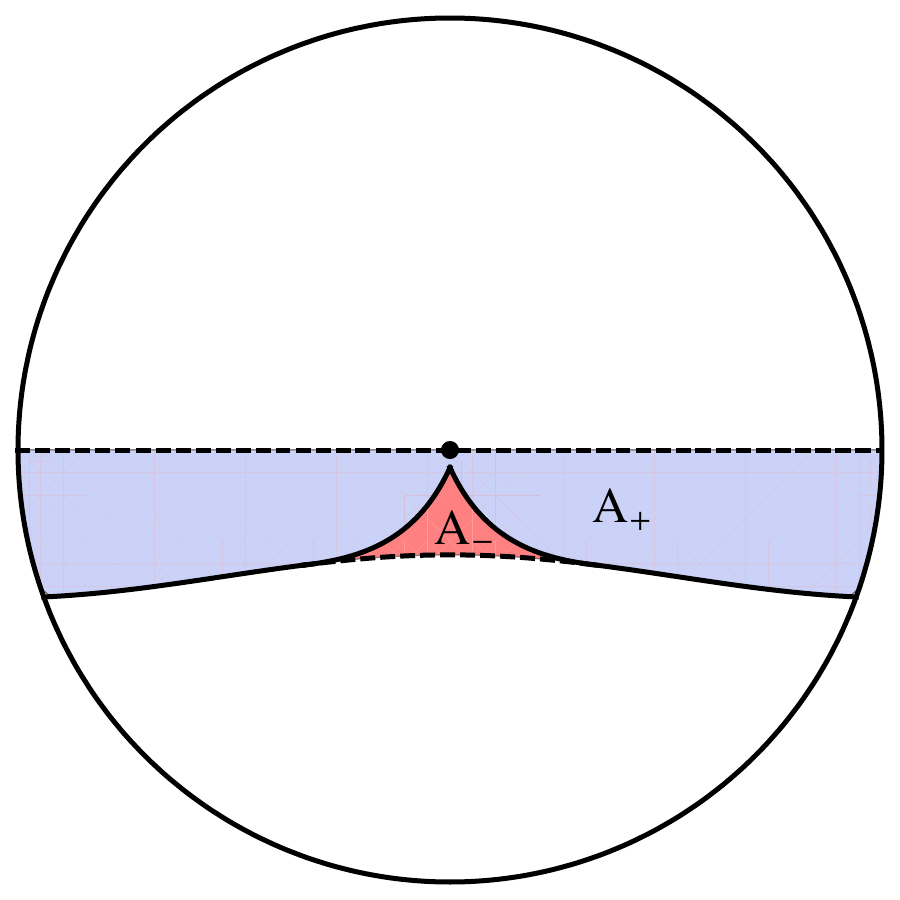}
\caption{Areas $A_+$ and $A_-$ corresponding to solid angles $S_+ = A_+/l^2$ and $S_- = A_-/l^2$.
The geometric phase correction due to the shape change is given by $S_-$}
\label{fig:Aminus}
\end{figure}

The total airborne time $T_{air}$ has small variability for platform diving, 
and is bounded above for springboard diving. A typical dive has $ 1.5 < T_{air} < 2.0$ seconds. 
After $E_t/l^2$ is determined by the choice of $m/n$ the airborne time can be adjusted 
by changing $l$ (within the physical possibilities) while keeping $E_t/l^2$ fixed.
Imposing $T_{air} = \tau_1 + \tau_5 + \tau_3$ we obtain:

\begin{corollary}
A twisting somersault with $m$ somersaults and $n$ twists in the kick-model must satisfy 
\[
     2 \pi m + n S = T_{air} \frac{ 2 E_s}{l} + 2 n l T_t \frac{ E_t - E_s}{l^2}
\]
where $T_{air}  - \tau_3  = \tau_1 + \tau_5 \ge 0$.
\end{corollary}

%

\section{The general twisting somersault}

The kick-model gives a good understanding of the principal ingredients needed in a successful dive. 
In the full model the shape-changing times $\tau_2$ and $\tau_4$ need to be set to realistic values.
We estimate that the full arm motion takes at least about $1/4$ of a second. So instead of having a kick 
connecting $\L_s$ to $\L_t(0)$, a piece of trajectory from the time-dependent Euler equations needs to be inserted, which can be seen in Fig.~\ref{fig:Lqfull}. 
The computation of the two dive segments from stage 2 and stage 4 has to be done numerically in general. 
Nevertheless, this is a beautiful generalisation of Montgomery's formula 
due to Cabrera \cite{Cabrera07}, which holds in the non-rigid situation.
In Cabrera's formula the geometric phase is still given by the solid angle enclosed by the 
trajectory, however for the dynamic phase instead of simply $2ET$ we actually need to integrate
$\L \cdot \O$ from 0 to $T$. Now when the body is rigid we have $2 E = \L \cdot \O = const$ and Cabrera's formula 
reduces back to Montgomery's formula.

\begin{theorem}
For the full model of a twisting somersault with $n$ twists,
the total amount of rotation $\Delta\Phi$ about the fixed angular momentum axis $\l$  is given by 
\[
  \Delta \phi = \Delta \phi_{kick} + \frac{ 2 \bar E_2 \tau_2 }{l}  + \frac{ 2 \bar E_4 \tau_4 }{l}  + S_-
\]
where $S_-$ is the solid angle of the triangular area on the $\L$-sphere enclosed by the trajectories of 
the shape-changing stage 2 and stage 4, and part of the trajectory of stage 3, 
see Fig.~\ref{fig:Aminus}.
The average energies along the transition segments are given by 
\[
   \bar E_i = \frac{1}{2\tau_i} \int_0^{\tau_i} \L \cdot \O \, \dee t, \quad i = 2, 4 
\]
\end{theorem}
\begin{proof}
This is a straightforward application of Cabrera's formula. For stage 1, stage 3, and stage 5 where there is no shape change
the previous formula is obtained. For stage 2 and stage 4 the integral of $\L \cdot \O$ along the numerically 
computed trajectory with time-dependent shape is computed to give the average energy during the shape change.
\end{proof}

\begin{figure}[ht]
\centering
\includegraphics[width=10cm]{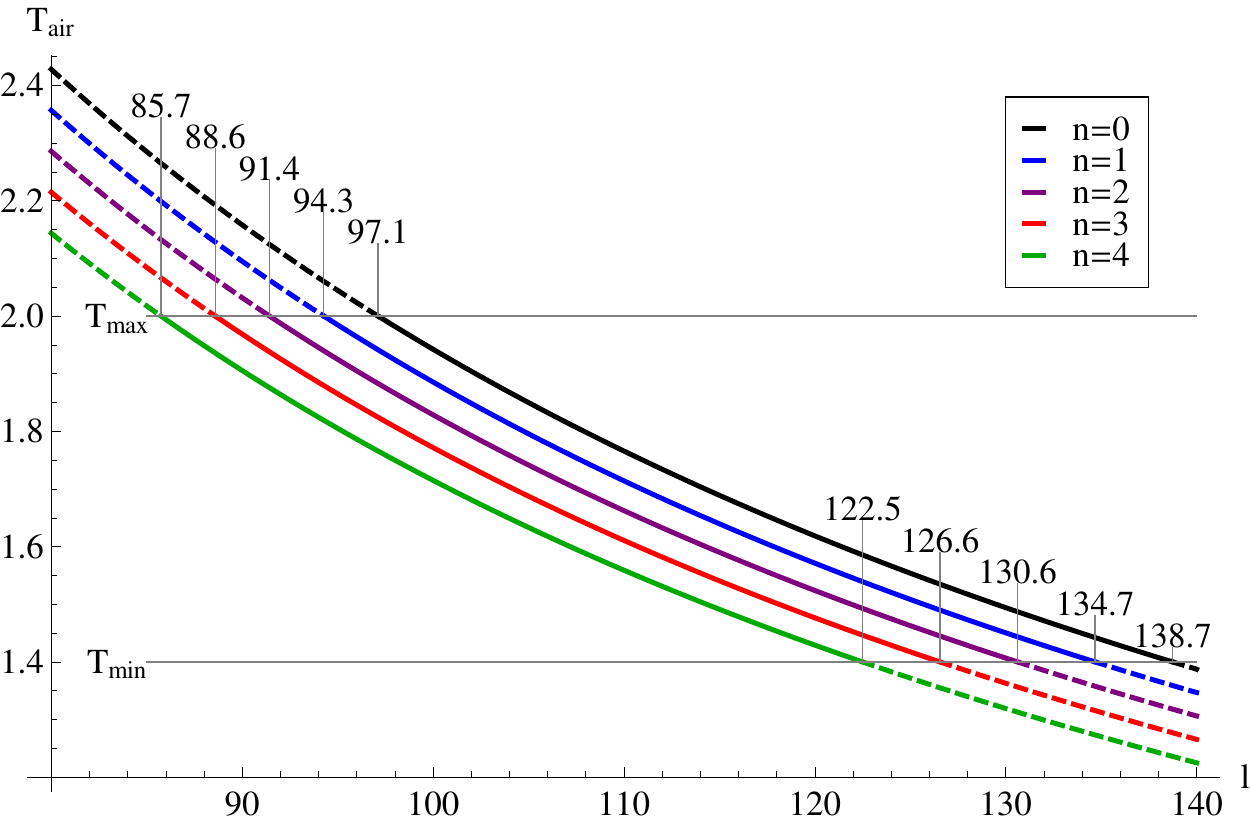}
\caption{The relationship between airborne time $T_\mathit{air}$ and angular momentum $l$ when $\tau_2=\tau_4=1/4$ is used in corollary~\ref{cor:fin}. 
The result is for the case of $m=1.5$ somersaults with different number of $n$ twists. 
 The maximum number of twists is $n=4$ since we need 
$T_{air} -   \tau_2 - \tau_3 - \tau_4  = \tau_1 + \tau_5 \ge 0$.
} \label{fig:tvlan}
\end{figure}

\begin{remark}
This quantifies the error that occurs with the kick-model. The geometric phase is corrected by the solid angle $S_-$ of a
small triangle, see Fig.~\ref{fig:Aminus}. The dynamic phase is corrected by adding terms proportional to $\tau_2$ 
and $\tau_4$. Note, if we keep the total time $\tau_2+ \tau_3 + \tau_4$ constant then we can think of the shape-changing times $\tau_2$ and $\tau_4$ from the full model as being part of the twisting somersault time $\tau_3$ of the kick-model. The difference is $2\Big((\bar{E}_2-E_t)\tau_2+(\bar{E}_4-E_t)\tau_4\big)/l$, and since both $\bar{E}_2$ and $\bar{E}_4 < E_t$ 
the dynamics phase in the full model is slightly smaller than in the kick-model.
\end{remark}

\begin{remark}
As $E_t$ is found using the endpoint of stage 2 it can only be calculated numerically now.
\end{remark}

The final step is to use the above results to find parameters that will achieve $m$ somersaults and $n$ twists, 
where typically $m$ is a half-integer and $n$ an integer. 
\begin{corollary}
A twisting somersault with $m$ somersaults and $n$ twists  satisfies 
\[
      2 \pi m + n S - S_- =  T_{air} \frac{ 2 E_s} { l} 
      +  2 \tau_2 \frac{\bar E_2 - E_s }{l} + 2 \tau_4 \frac{ \bar E_4 - E_s}{l}  + 2 \tau_3 \frac{ E_t - E_s}{ l}
\]
where $T_{air} -   \tau_2 - \tau_3 - \tau_4  = \tau_1 + \tau_5 \ge 0$.
\label{cor:fin}
\end{corollary}

\begin{remark}
Even though $\bar E_2, \bar E_4, E_t$, and $S_-$ have to be computed numerically in this formula, the geometric interpretation is as clear as before: The geometric phase is given by the area terms $nS$ and $S_-$. 
\end{remark}

In the absence of explicit solutions for the shape-changing stages 2 and 4, we have numerically evaluated 
the corresponding integrals and compared the predictions of the theory to a full numerical simulation. 
The results for a particular case and parameter scan are shown in Fig.~\ref{fig:Lqfull} and Fig.~\ref{fig:tvlan} respectively, and the agreement between theory and numerical simulation is extremely good.
Fixing the shape change and the time it takes determines $E_t/l^2$, so the essential parameters to be adjusted by the athlete are the angular momentum $l$ and airborne time $T_{air}$ (which are directly related to the initial angular and vertical velocities at take-off).
Our result shows that these two parameters are related in a precise way given in Corollary~\ref{cor:fin}.
At first it may seem counterintuitive that a twisting somersault with more twists (and same number of somersaults) requires less angular momentum when the airborne time is the same, as shown in Fig.~\ref{fig:tvlan} for $m = 3/2$ and $n = 0,1,2,3,4$. The reason is that while twisting, the moments of inertia relevant for somersaulting are smaller than not twisting, since pure somersaults take layout position as shown in Fig.~\ref{fig:animationlayout}, hence less overall time is necessary. 
In reality, the somersaulting phase is often done in pike or tuck position which significantly reduces the moment of inertia about the somersault axis, leading to the intuitive result that more twists require larger angular momentum when airborne time is the same.

\section{Acknowledgement}

This research was supported by ARC Linkage grant LP100200245 and the New South Wales Institute of Sports.

\bibliographystyle{plain}
\bibliography{../bib_cv/all,../bib_cv/hd,../bib_cv/SoSa}

\def\cprime{$'$}
\begin{thebibliography}{10}

\bibitem{Cushman05}
Larry Bates, Richard Cushman, and Emil Savev.
\newblock The rotation number and the herpolhode angle in euler's top.
\newblock {\em Zeitschrift f{\"u}r Angewandte Mathematik und Physik (ZAMP)},
  56(2):183--191, 2005.

\bibitem{BDDLT15}
Sudarsh Bharadwaj, Nathan Duignan, Holger~R. Dullin, Karen Leung, and William
  Tong.
\newblock The diver with a rotor.
\newblock {\em http://arxiv.org/abs/1510.02978}, 2015.

\bibitem{Cabrera07}
Alejandro Cabrera.
\newblock A generalized montgomery phase formula for rotating self-deforming
  bodies.
\newblock {\em Journal of Geometry and Physics}, 57:1405--1420, 2007.

\bibitem{Enos93}
Michael~J. Enos.
\newblock On an optimal control problem on {${\rm SO}(3)\times{\rm SO}(3)$} and
  the falling cat.
\newblock In {\em Dynamics and control of mechanical systems (Waterloo, ON,
  1992)}, volume~1 of {\em Fields Inst. Commun.}, pages 75--111. Amer. Math.
  Soc., Providence, RI, 1993.

\bibitem{Frohlich79}
C.~Frohlich.
\newblock Do springboard divers violate angular-momentum conservation?
\newblock {\em Am. J. Phys.}, 47(7):583--592, 1979.

\bibitem{Iwai98}
Toshihiro Iwai.
\newblock The mechanics and control for multi-particle systems.
\newblock {\em J. Phys. A}, 31:3849--3865, 1998.

\bibitem{Iwai99}
Toshihiro Iwai.
\newblock Classical and quantum mechanics of jointed rigid bodies with
  vanishing total angular momentum.
\newblock {\em J. Math. Phys.}, 40(5):2381--2399, 1999.

\bibitem{Levi93}
Mark Levi.
\newblock Geometric phases in the motion of rigid bodies.
\newblock {\em Archive for rational mechanics and analysis}, 122(3):213--229,
  1993.

\bibitem{Montgomery91}
Richard Montgomery.
\newblock How much does the rigid body rotate? {A} {B}erry's phase from the
  18th century.
\newblock {\em Am. J. Phys}, 59(5):394--398, 1991.

\bibitem{Montgomery93}
Richard Montgomery.
\newblock Gauge theory of the falling cat.
\newblock In {\em Dynamics and control of mechanical systems ({W}aterloo, {ON},
  1992)}, volume~1 of {\em Fields Inst. Commun.}, pages 193--218. Amer. Math.
  Soc., Providence, RI, 1993.

\bibitem{PD12}
G.~Papadopoulos and H.~R. Dullin.
\newblock Semi-global symplectic invariants of the {E}uler top.
\newblock {\em Journal of Geometric Mechanics}, 5(2):215--232, 2013.

\bibitem{ShapereWilczek87}
A.~Shapere and F.~Wilczek.
\newblock Self-propulsion at low reynolds number.
\newblock {\em Phys. Rev. Lett.}, 58:2051--2054, 1987.

\bibitem{ShapereWilczek88}
Alfred Shapere and Frank Wilczek.
\newblock Gauge kinematics of deformable bodies.
\newblock {\em Am. J. Phys.}, 57(6):514--518, 1989.

\bibitem{Tong15}
William Tong.
\newblock {\em Coupled Rigid Body Dynamics with Application to Diving}.
\newblock PhD thesis, University of Sydney, 2015.

\bibitem{Yeadon90b}
M.~R. Yeadon.
\newblock The simulation of aerial movement{ - II}. {A} mathematical inertia
  model of the human body.
\newblock {\em Journal of Biomechanics}, 23:67--74, 1990.

\bibitem{Yeadon93a}
M.~R. Yeadon.
\newblock The biomechanics of twisting somersaults: Part {I}. rigid body
  motions.
\newblock {\em Journal of Sports Sciences}, 11:187--198, 1993.

\bibitem{Yeadon93b}
M.~R. Yeadon.
\newblock The biomechanics of twisting somersaults: Part {II}. contact twist.
\newblock {\em Journal of Sports Sciences}, 11:199--208, 1993.

\bibitem{Yeadon93c}
M.~R. Yeadon.
\newblock The biomechanics of twisting somersaults: Part {III}. aerial twist.
\newblock {\em Journal of Sports Sciences}, 11:209--218, 1993.

\bibitem{Yeadon93d}
M.~R. Yeadon.
\newblock The biomechanics of twisting somersaults: Part {IV}. partitioning
  performances using the tilt angle.
\newblock {\em Journal of Sports Sciences}, 11:219--225, 1993.

\end{thebibliography}

\end{document}